\documentclass[a4paper]{llncs}
\usepackage{amsmath,amssymb}
\usepackage{graphicx}
\newtheorem{observation}[theorem]{Observation}

\newcommand{\type}{\mathit{type}}
\newcommand{\flip}{\mathit{flip}}
\renewcommand{\L}{\mathbb{L}}
\newcommand{\T}{\mathbb{T}}
\newcommand{\F}{\mathbb{F}}

\newcommand{\A}{\mathbb{A}}
\newcommand{\G}{\mathbb{G}}
\newcommand{\Q}{\mathbb{Q}}
\newcommand{\Id}{\mathit{Id}}
\renewcommand{\S}{\mathbb{S}}

\def\qed{\ifhmode\unskip\nobreak\fi\hfill\ifmmode\square\else$\square$\fi}

\def\ColoredCube{\textsc{Colored Cube Isomorphism}}

\def\GI{\mathsf{GI}}

\def\ComplexityTheorem{The problem {\ColoredCube} is $\GI$-complete even if both input colorings has a form $n^2 \rightarrow [2]$.}

\title{Automorphisms of the Cube $n^d$}

\renewcommand{\problem}[3]{\vskip 0.1cm \begin{tabular}{lp{9.5cm}} PROBLEM: &  {\sc #1} \\ {\it Instance}: & #2\\ {\it Question}: & #3\end{tabular} \vskip 0.1cm}

\author{
Pavel Dvo\v r\'ak\inst{1}
\thanks{The research leading to these results has received funding from the European Research Council under the European Union's Seventh Framework Programme (FP/2007-2013) / ERC Grant Agreement n. 616787. Supported by Czech Science Foundation 
GA{\v C}R (grant \#19-27871X).}
\and
Tom\'a\v s Valla\inst{2}
\thanks{T. Valla acknowledges the support of the OP VVV MEYS funded project
 10 CZ.02.1.01/0.0/0.0/16\_019/0000765 ``Research Center for Informatics''.}
}

\institute{
Faculty of Mathematics and Physics,\\
Charles University in Prague, Czech Republic\\
\email{koblich@iuuk.mff.cuni.cz}
\medskip
\and
Faculty of Information Technology,\\
Czech Technical University in Prague, Czech Republic\\
\email{tomas.valla@fit.cvut.cz}
}

\begin{document}

\maketitle

\begin{abstract}
Consider a hypergraph $n^d$ where the vertices are points
of the $d$-dimensional cube $[n]^d$ and the edges
are all sets of $n$ points such that they are in one line.
We study the structure of the group of automorphisms of $n^d$,
i.e., permutations of points of $[n]^d$ preserving the edges.
In this paper we provide a complete characterization.
Moreover, we consider the {\ColoredCube} problem
of deciding whether for two colorings of the vertices of $n^d$
there exists an automorphism of $n^d$ preserving the colors.
We show that this problem is $\GI$-complete.
\footnote{A preliminary version of this work  appeared in
the proceedings of Computing and Combinatorics -- 22nd International Conference, COCOON 2016.}
\end{abstract}

\section{Introduction}
\label{sec:Intro}

Let us denote $[n] = \{1, \dots, n\}$.
Let $[n]^d$ be the set of all points $(p_1,\dots,p_d)$ such that $p_i\in[n]$ for every $1\le i\le d$.
Let $s = (s^1, \dots ,s^n)$ be a sequence of $n$ distinct points of $[n]^d$.
Let $s^i = [s^i_1,\dots,s^i_d]$ for every $1 \leq i \leq n$.
We say that $s$ is \emph{linear} if for every $1 \leq j \leq d$ a sequence $\tilde{s}_j = (s^1_j, \dots, s^n_j)$
is strictly increasing, strictly decreasing or constant.
Note that at least one sequence $\tilde{s}_j$ is nonconstant as $s$ is a sequence of $n$ distinct points.
A set of points $\{p^1, p^2, \dots , p^n\}\subseteq n^d$ is a \emph{line} if it can be ordered into a linear sequence $(q^1, q^2, \dots, q^n)$.
We denote the set of all lines of $[n]^d$ by $\L(n^d)$.
A \emph{combinatorial cube} $n^d$ is a hypergraph $\bigl([n]^d, \L(n^d)\bigr)$.
Note that there is a fundamental difference between the combinatorial cube $n^d$
and another well-studied structure, the \emph{hypercube} $Q_d$,
defined as the graph $Q_d=\bigl(\{0,1\}^d,E\bigr)$ where $\{u,v\}\in E$ if and only if the vectors $u,v$
differ in exactly one coordinate.

We denote the group of all permutations on $n$ elements by $\mathbb{S}_n$.
A permutation $S \in \mathbb{S}_{n^d}$ is an \emph{automorphism} of the cube $n^d$ if
$\ell = \{v_1, \dots ,v_n\} \in \L(n^d)$ implies $S(\ell) = \bigl\{S(v_1), \dots ,S(v_n)\bigr\} \in \L(n^d)$.
Informally, an automorphism of the cube $n^d$ is a permutation of the cube points
which preserves the lines.
We denote the set of all automorphisms of $n^d$ by $T_n^d$.
Note that all automorphisms of $n^d$
with a composition $\circ$ form a group $\T_n^d = (T_n^d, \circ, \Id)$.
Our main result is the characterization of the generators of the group $\T_n^d$ and
computing the order of $\T_n^d$.
Surprisingly, the structure of  $\T_n^d$ is richer than only the obvious rotations and symmetries.
We use two groups of automorphisms for characterization of the group $\T_n^d$ as follows.

The first is a group $\G_n^d$ which is isomorphic to the hypercube automorphism group $\Q_d$~\cite{harary99}.
Generators of $\Q_d$ are
\begin{enumerate}
\item \emph{Translations} $T_a$ by $a \in \{0,1\}^d$, $T_a\bigl([x_1,\dots,x_d]\bigr) = \bigl[x_1 + a_1,\dots,x_d + a_d\bigr]$ where the sum is modulo 2.
 \item \emph{Rotations} $R_\pi$ by $\pi \in \S_d$, $R_\pi\bigl([x_1,\dots,x_d]\bigr) = [x_{\pi(1)},\dots,x_{\pi(d)}]$.
\end{enumerate}
It is known that every automorphism of the hypercube can be composed as $T \circ R$ for a translation $T$ and a rotation $R$.
To use automorphisms in $\Q_d$ for the combinatorial cube, we need to change the definition of the translations.
The rotations can be used immediately. 
Thus, the group $\G_n^d$ is generated by
\begin{enumerate}
 \item \emph{Translations} $T_a$ by $a \in \{0,1\}^d$, $T_a\bigl([x_1,\dots,x_d]\bigr) = \bigl[\flip(x_1,a_1),\dots,\flip(x_d,a_d)\bigr]$ where
 \[
  \flip(i,b) = 
  \begin{cases}
   i & b = 0, \\
   n - i + 1 & b = 1.
  \end{cases}
 \]
 \item \emph{Rotations} $R_\pi$ by $\pi \in \S_d$, $R_\pi\bigl([x_1,\dots,x_d]\bigr) = [x_{\pi(1)},\dots,x_{\pi(d)}]$.
\end{enumerate}
It is clear that groups $\G_n^d$ and $\Q_d$ are isomorphic and thus every automorphism in $\G_n^d$ can be composed as $T \circ R$ for a translation $T$ and a rotation $R$.

The second group is a group of permutation automorphisms $\mathbb{F}_{n}$ that contains mappings
$F_\rho \bigr([x_1, \dots ,x_d]\bigl) = \bigl[\rho(x_1), \dots ,\rho(x_d)\bigr]$
where $\rho \in \mathbb{S}_n$ such that it has a \emph{symmetry property}:
if $\rho(i) = j$ then $\rho(n - i +1) = n - j + 1$.

Our main result is summarized in the following theorem. For the proof we use and generalize some ideas of Silver~\cite{silver67},
who characterized the group of automorphisms of the cube $4^3$.
\begin{theorem}
\label{thm:MainResult}
Let $n > 2$.
The group $\T_n^d$ is generated by the elements of $\G_n^d \cup \F_n$.
The order of the group $\T_n^d$ is $2^{d-1 + k} d!k!$ where $k = \lfloor\frac{n}{2}\rfloor$.
\end{theorem}

An \emph{isomorphism} of two hypergraphs $H_1 = (V_1, E_1)$, $H_2 = (V_2, E_2)$
is a bijection $f:V_1\to V_2$ such that for each $\{v_1,\dots,v_r\}\subseteq V_1$,
$\{v_1,\dots,v_r\}\in E_1 \Leftrightarrow \{f(v_1),\dots,f(v_r)\}\in E_2$.
A \emph{coloring} of a hypergraph $H = (V, E)$ by $k$ colors is a function $s: V \to [k]$.
The following problem is well studied.

\problem{Colored Hypergraph Isomorphism (CHI)}
  {Hypergraphs $H_1 = (V_1, E_1), H_2 = (V_2, E_2)$, colorings ${s_1: V_1 \to [k]}, {s_2: V_2 \to [k]}$.}
  {Is there an isomorphism $f: V_1 \to V_2$ of $H_1$ and $H_2$ such that it preserves the colors? I.e., it holds $s_1(v) = s_2\bigl(f(v)\bigr)$ for every vertex $v$ in $V_1$.}

There are several {\sf FPT} algorithms\footnote{The parameter is the maximum number of vertices colored by the same color.} for {\sc CHI}---see Arvind et. al.~\cite{arvind06,arvind13}.
The problem {\ColoredCube} is defined as the problem {\sc CHI} where both $H_1, H_2= n^d$.
Since we know the structure of the group $\T^d_n$, it is natural to ask if {\ColoredCube} is an easier problem than {\sc CHI}.
We prove that the answer is negative.
The class of decision problems ${\GI}$ contains all problems with a polynomial
reduction to the problem {\sc Graph Isomorphism}.

\problem{Graph Isomorphism}
  {Graphs $G_1, G_2$.}
  {Are the graphs $G_1$ and $G_2$ isomorphic?}

\noindent

It is well known that {\sc CHI} is ${\GI}$-complete, see Booth and Colbourn~\cite{booth77}.
We prove the same result for {\ColoredCube}.




\begin{theorem}
\label{thm:Complexity}
\ComplexityTheorem
\end{theorem}


The paper is organized as follows.
In Section~\ref{sec:Prelim} we present some basic properties of the combinatorial cube $n^d$, prove that $\G_n^d$ and $\F_n$ are automorphism groups and also we count the order of the group $\T_2^d$,
which structure is different from other automorphism groups.
Next in Sections~\ref{sec:Technical} and~\ref{sec:GeneralCube}, we characterize the generators for $\T_n^d$.
In Section~\ref{sec:SizeTnd} we count the order of the group $\T_n^d$.
In the last section we study the complexity of {\ColoredCube} and prove Theorem~\ref{thm:Complexity}.

\subsection{Motivation}

A natural motivation for this problem comes from the game of Tic-Tac-Toe.
It is usually played on a 2-dimensional square grid and each player puts his
tokens (usually crosses for the first player and rings for the second)
at the points on the grid. A player wins if he occupies a line with his token
vertically, horizontally or diagonally (with the same length as the grid size)
faster than his opponent.
Tic-Tac-Toe is a member of a large class of games called strong positional
games. For an extraordinary reference see Beck~\cite{beck06}.
The size of a basic Tic-Tac-Toe board is $3 \times 3$ and it is easy to show
by case analysis that the game is a draw if both players play optimally.
However, the game can be generalized to a larger grid and more dimensions.
The $d$-dimensional Tic-Tac-Toe is played on the points of a $d$-dimensional combinatorial cube
and it is often called the game $n^d$.
With larger boards the case analysis becomes unbearable even using computer
search and clever algorithms have to be devised.

The only (as far as we know) non-trivial solved 3-dimensional Tic-Tac-Toe is
the game $4^3$, which is called Qubic. Qubic is a win for the first player,
which was shown by Patashnik~\cite{patashnik80} in 1980.
It was one of the first examples of computer-assisted proofs based
on a brute-force algorithm, which utilized several clever techniques for pruning the game tree.
Another remarkable approach for solving Qubic was made by Allis~\cite{allis94} in 1994,
who introduced several new methods.
However, one technique is common for both authors: the detection of isomorphisms of game configurations.
As the game of Qubic is highly symmetric, this detection
substantially reduces the size of the game tree.

For the game $n^d$, theoretical results are usually achieved for large $n$ or large $d$.
For example, by the famous Hales and Jewett theorem~\cite{hales63}, for any $n$
there is (an enormously large) $d$ such that the hypergraph $n^d$ is not 2-colorable,
that means, the game $n^d$ cannot end in a draw.
Using the standard Strategy Stealing argument, $n^d$ is thus a first player's win.
In two dimensions, each game $n^2$, $n>2$, is a draw (see Beck~\cite{beck06}).
Also, several other small $n^d$ are solved.

All automorphisms for Qubic were characterized by Rolland Silver~\cite{silver67} in 1967.
As in the field of positional games the game $n^d$ is intensively studied
and many open problems regarding $n^d$ are posed,
the characterization of the automorphism group of $n^d$ is a natural task.

The need to characterize the automorphism group came from our real
effort to devise an algorithm and computer program that would be
able to solve the game $5^3$, which is the smallest unsolved Tic-Tac-Toe game.
While our effort of solving $5^3$ is currently not yet successful, we were
able to come up with the complete characterization of the automorphism group $n^d$,
giving an algorithm for detection of isomorphic positions not only in the game $5^3$,
but also in $n^d$ in general.

A game configuration can be viewed as a coloring $s$ of $n^d$ by crosses,
rings and empty points, i.e., $s: n^d \rightarrow [3]$.
Since we know the structure of the group $\T^d_n$, this characterization yields
an algorithm for detecting isomorphic game positions by simply trying all combinations
of the generators (the number of the combinations is given by the order of the group $\T^d_n$).
A natural question arises: can one obtain a faster algorithm?
Note that the hypergraph $n^d$ has polynomially many edges in the number of vertices.
Therefore, from a point of view of polynomial-time algorithms it does not matter if there are hypergraphs $n^d$ with colorings or only colorings on the input of {\ColoredCube}.
Due to Theorem~\ref{thm:Complexity} we conclude that deciding if two game
configurations are isomorphic is as hard as deciding if two graphs are isomorphic.

Although our primary motivation came from the game of Tic-Tac-Toe,
we believe our result has much broader interest
as it presents an analogy of automorphism characterization results of
hypercubes (see e.g. \cite{choudum08,harary99}).

\section{Preliminaries}
\label{sec:Prelim}
Beck~\cite{beck06} in his work defined lines to be ordered (the linear sequences in our case).
However, for us it is more convenient to have unordered lines because some automorphisms will change the order of points in the line.

Let $\ell$ be a line and $q = (q^1,\dots,q^n)$ be an ordering of $\ell$ into a linear sequence.
Note that every line in $\L(n^d)$ has two such orderings. 
Another ordering of $\ell$ into a linear sequence is $(q^n,\dots,q^1)$.
We define a \emph{type} of a sequence $\tilde{q}_j = (q^1_j,\dots,q^n_j)$ as
$+$ if $\tilde{q}_j$ is strictly increasing,
$-$ if $\tilde{q}_j$ is strictly decreasing,
$c$ if $\tilde{q}_j$ is constant and $q^i_j = c$ for every $1 \leq i \leq n$.
A type of $q$ is $\type(q) = \bigl(\type(\tilde{q}_1),\dots,\type(\tilde{q_n})\bigr)$.

\emph{Type} of a line $\ell$ is a type of an ordering $q$ of $\ell$ into a linear sequence such that the first non-constant entry of $\type(q)$ is $+$.
For example, let $$\ell = \bigl\{[1,1,4], [1,2,3], [1,3,2], [1,4,1]\bigr\} \in \L(4^3)$$ and $q_1$ and $q_2$ be distinct orderings of $\ell$ into a linear sequence.
Then, $ \type(q_1) = (1,+,-)$ and $\type(q_2) = (1,-,+)$.
By definition $\type(\ell) = \type(q_1) = (1,+,-)$.

Let us now define several terms we use in the rest of the paper.
A \emph{dimension $\dim(\ell)$} of a line $\ell \in \L(n^d)$ is
$
\dim(\ell) = \bigl|\bigl\{i \in \{1,\dots,d\}| \type(\ell)_i \in \{+,-\}\bigr\}\bigr|.
$
A \emph{degree $\deg(p)$} of a point $p \in [n]^d$ is a number of incident lines, formally
$
\deg(p) = \bigl|\{\ell \in \L(n^d)| p \in \ell\}\bigr|.
$
Two points $p_1, p_2 \in [n]^d$ are \emph{collinear}, if there exists a line $\ell \in \L(n^d)$, such that $p_1 \in \ell$ and $p_2 \in \ell$.
A point $p \in [n]^d$ is called a \emph{corner} if $p$ has coordinates only $1$ and $n$.
A point $p = [x_1, \dots, x_d] \in [n]^d$ is an \emph{outer point} if there exists at least one $i \in \{1, \dots, d\}$ such that $x_i \in \{1, n\}$.
If a point $p \in [n]^d$ is not an outer point then $p$ is called an \emph{inner point}.

A line $\ell \in \L(n^d)$ is called an \emph{edge} if $\dim(\ell)=1$ and $\ell$ contains two corners.
Two corners are \emph{neighbors} if they are connected by an edge.
A line $\ell \in \L(n^d)$ with $\dim(\ell) = d$ is called a \emph{main diagonal}.
We denote the set of all main diagonals by $\L_m(n^d)$.
For better understanding the notions see Figure~\ref{fig:example} with some examples in the cube $4^3$.
\begin{figure}
\centering
\includegraphics{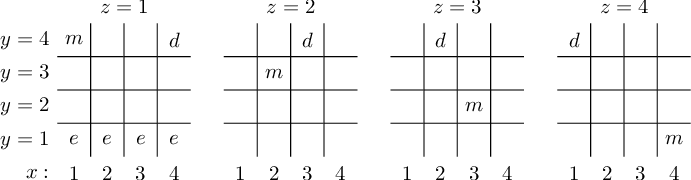}
\caption{The cube $4^3$ with some examples of lines. An edge $e$ has a type $(+,1,1)$, a line $d$ has a dimension 2 and a type $(+,4,-)$ and a main diagonal $m$ has a type $(+,-,+)$.}
\label{fig:example}
\end{figure}

A $k$-dimensional \emph{face} $F$ of the cube $n^d$ is a maximal set of points of $n^d$,
such that there exist two index sets $I,J \subseteq \{1,\dots,d\}, I \cap J = \emptyset, |I| + |J| = d - k$
and for each point $[x_1, \dots, x_d]$ in $F$ holds that
$x_i=1$ for each $i \in I$ and, $x_j = n$ for each $j \in J$.
For example, $\bigl\{[x,y,1,n]|x,y \in [n] \bigr\}$ is a 2-dimensional face of the cube $n^4$.
Note that an edge is an $1$-dimensional face.

A point $p \in [n]^d$ is \emph{fixed} by an automorphism $S$ if $S(p)=p$.
A set of points $\{p_1, \dots ,p_k\}$ is fixed by an automorphism $S$
if $\{p_1, \dots ,p_k\} = \bigl\{S(p_1), \dots ,S(p_k)\bigr\}$.
Note that if a set $B$ is fixed it does not necessarily mean every point of $B$ is fixed.

For $n$ odd, we denote $\gamma = \frac{n+1}{2}$ and the \emph{center} of the cube $n^d$ is the point $c = [\gamma, \dots, \gamma]$.

\subsection{Order of $\T^d_2$}

The combinatorial cube $2^d$ is different from the other cubes because every two points are collinear.
Thus, we have the following proposition.
\begin{proposition}
The order of the group $\T^d_2$ is $(2^d)!$.
\end{proposition}
\begin{proof}
Every permutation of the points of the cube $2^d$ is an automorphism,
as the graph $2^d$ is the complete graph on $2^d$ vertices.
\qed\end{proof}
\noindent We further assume that $n > 2$.

\subsection{Basic Groups}
In this subsection we prove that the basic groups $\F_n$ and $\G_n^d$ are groups of automorphisms of the combinatorial cube $n^d$.
In the following proofs we use $\ell$ as an arbitrary line and $q = (q^1,\dots,q^n)$ as an ordering of $\ell$ into a linear sequence.
\begin{lemma}
 Every $F_\rho \in \F_n$ is an automorphism of the combinatorial cube $n^d$.
\end{lemma}
\begin{proof}
 We recall that $F_\rho\bigl([x_1,\dots,x_d]\bigr) = \bigl[\rho(x_1),\dots,\rho(x_d)\bigr]$ where $\rho \in \S_n$.
 Let $\sigma = \rho^{-1}$ and 
 \[
 p = \Bigl(\bigl[\rho\bigl(q^{\sigma(1)}_1\bigr),\dots,\rho\bigl(q^{\sigma(1)}_d\bigr)\bigr],\dots,\bigl[\rho\bigl(q^{\sigma(n)}_1\bigr),\dots,\rho\bigl(q^{\sigma(n)}_d\bigr)\bigr]\Bigr). 
 \]
 We claim that $p$ is an ordering of $F_\rho(\ell)$ into a linear sequence.
 Consider sequences of the $j$-th coordinations of $q$ and $p$. Thus,
 \begin{align*}
  \tilde{q}_j &= \bigl(q^1_j,\dots,q^n_j\bigr),\\
  \tilde{p}_j &= \Bigl(\rho\bigl(q^{\sigma(1)}_j\bigr),\dots,\rho\bigl(q^{\sigma(n)}_j\bigr)\Bigr).  
 \end{align*}
 If $\type(\tilde{q}_j) = c$ then clearly $\type(\tilde{p}_j) = \rho(c)$.
 If $\type(\tilde{q}_j) = +$ then $q^i_j = i$ and
 \[
  \Bigl(\rho\bigl(q^{\sigma(1)}_j\bigr),\dots,\rho\bigl(q^{\sigma(n)}_j\bigr)\Bigr) = \Bigl(\rho\bigl(\sigma(1)\bigr),\dots,\rho\bigl(\sigma(n)\bigr)\Bigr) = (1,\dots,n).
 \]
 Thus, $\type(\tilde{p}_j) = +$.
 In the last case, if $\type(\tilde{q}_j) = -$ then $q^i_j = n - i + 1$. 
 In this case we use that $\rho$ has the symmetry property $\bigl(\rho(n - i + 1) = n - \rho(i) + 1\bigr)$.
 For all $i \in [n]$ holds that
 \[
  \rho\bigl(q^{\sigma(i)}_j\bigr) = \rho \bigl(n - \sigma(i) + 1\bigr) = n - i + 1.
 \]
 Thus,$ \bigl(\rho\bigl(q^{\sigma(1)}_j\bigr),\dots,\rho\bigl(q^{\sigma(n)}_j\bigr)\bigr) = (n,\dots,1)$ and $\type(\tilde{p}_j) = -$.
 Therefore, we prove that $F_\rho(\ell)$ is a line.
 \qed
\end{proof}

\begin{lemma}
 Every translation $T_a \in \G_n^d$ is an automorphism of the combinatorial cube $n^d$.
\end{lemma}
\begin{proof}
We recall that $T_a\bigl([x_1\dots,x_d]\bigr) = \bigl[\flip(x_1,a_1),\dots,\flip(x_d,a_d)\bigr]$, where
\[
 \flip(i,b) = 
 \begin{cases}
  i & b = 0,\\
  n - i + 1 & b = 1.
 \end{cases}
\]
Let $p = \bigl(T_a(q^1),\dots,T_a(q^n)\bigr)$.
We prove that $p$ is an ordering of $T_a(\ell)$ into a linear sequence.
Consider a sequence $\tilde{q}_j$ (or $\tilde{p}_j$) of the $j$-th coordinates of $q$ (or $p$).
If $\type(\tilde{q}_j) = c$ then $\type(\tilde{p}_j) = c$ if $a_j = 0$ or $\type(\tilde{p}_j) = n - c + 1$ if $a_j = 1$.
If $\type(\tilde{q}_j) = +$ then $p_j^i = i$ if $a_j = 0$ or $p_j^i = n - i + 1$ if $a_j = 1$.
Thus,
\[
 \bigl(T_a(q^1)_j,\dots,T_a(q^n)_j\bigr) = 
 \begin{cases}
  (1,\dots,n) & a_j = 0, \\
  (n,\dots,1) & a_j = 1.
  \end{cases}
\]
Thus, $\type(\tilde{p}_j) = +$ or $-$ depending on $a_j$.
If $\type(\tilde{q}_j) = -$ the situation is opposite.
If $a_j = 0$ then $p^i_j = n - i - 1$ and $p^i_j = i$ if $a_j = 1$.
Thus, again $\type(\tilde{p}_j) = +$ or $-$.
 \qed
\end{proof}

\begin{lemma}
 Every rotation $R_\pi \in \G_n^d$ is an automorphism of the combinatorial cube $n^d$.
\end{lemma}
\begin{proof}
We recall that $R_\pi\bigl([x_1\dots,x_d]\bigr) = [x_{\pi(1)},\dots,x_{\pi(d)}]$, where $\pi \in \S_d$.
We claim that $p = \bigl(R_\pi(q^1),\dots,R_\pi(q^n)\bigr)$ is an ordering of $R_\pi(\ell)$ into a linear sequence.
Let $\tilde{q}_j$ and $\tilde{p}_j$ be sequences of $j$-th coordinates of $q$ or $p$, respectively.
Let $\sigma = \pi^{-1}$. 
Note that the sequence $\tilde{p}_j$ is exactly the sequence $\tilde{q}_{\sigma(j)}$.
Thus, every sequence $\tilde{p}_j$ has a type $+$, $-$ or a constant.
\qed
\end{proof}

\section{Corners, Main Diagonals and Edges}
\label{sec:Technical}

In this section we investigate how every automorphism $S \in \T_n^d$ maps main diagonals, edges and corners.
First, we prove some easy observation, which were also used by Silver~\cite{silver67}.

\begin{observation}
\label{obs:FixedLine}
If an automorphism $S \in \T_{n}^d$ fixes two collinear points $p, q \in [n]^d$,
then $S$ also fixes a line $\ell \in \L(n^d)$ such that $p,q \in \ell$.
\end{observation}
\begin{proof}
For any two distinct points $p_1, p_2 \in [n]^d$ there is at most one line $\ell \in \L(n^d)$ such that $p_1, p_2 \in \ell$.
Therefore, if the points $p$ and $q$ are fixed then the line $\ell$ has to be fixed as well.
\qed\end{proof}

\begin{observation}
\label{obs:FixedIntersection}
If two lines $\ell_1, \ell_2 \in \L(n^d)$ are fixed by $S \in \T_{n}^d$
then their intersection, a point $p$ in $\ell_1 \cap \ell_2$, is fixed by $S$.
\end{observation}
\begin{proof}
For any two lines $\ell, \ell'$ there is at most one point in $\ell \cap \ell'$. 
Therefore, if the lines $\ell_1$ and $\ell_2$ are fixed then the point $p$ has to be fixed as well.
\qed\end{proof}

\begin{lemma}
\label{lem:FrontLine}
Let $F  = \bigr\{[x,y,1, \dots ,1]|x,y \in [n]\bigl\}$ be a 2-dimensional face of $n^d$,
and let an automorphism $S \in \T_{n}^d$ fixes all 4 corners of $F$,
i.e., points $[1,\dots,1]$, $[n,1, \dots ,1]$, $[1, n, 1, \dots ,1]$ and $[n, n,1, \dots ,1]$.
Then, if $S$ fixes a point $[i,1, \dots ,1], i \in [n]$ it also fixes a point $[n-i+1,1, \dots ,1]$.
\end{lemma}

\begin{proof}
The automorphism $S$ fixes all 4 corners of $F$, therefore by Observation~\ref{obs:FixedLine}, it fixes both diagonals $d_1, d_2 \subset F$.
The types of $d_1$ and $d_2$ are $\type(d_1) = (+,+,1, \dots ,1)$ and $\type(d_2) = (+,-,1, \dots ,1)$.

Suppose that $S$ fixes a point $p = [i, 1, \dots ,1]$, where $i \in \{2,\dots, n-1\}$ (corners are already fixed).
In three steps we show that the point $p_5 = [n-i+1,1, \dots ,1]$ is fixed (note that for $i = n-i+1$ the proof is trivial, thus we suppose $i \neq \gamma$ for odd $n$). 
Fixed points in a face $7 \times 7$ are depicted in Figure~\ref{fig:FixFrontFace}.

\begin{figure}[h]
\centering
\includegraphics{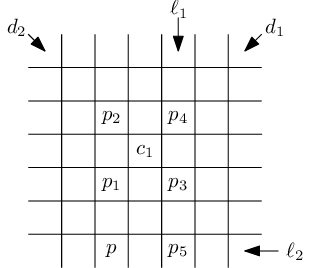}
\caption{How to fix points by diagonals in a front 2-dimensional face.}
\label{fig:FixFrontFace}
\end{figure}

First we show that $p_1 = [i,i,1, \dots, 1]$ is fixed by $S$.
A point $S(p_1)$ must be on $d_1$ and it must be collinear with $p$.
There are 2 points collinear with $p$ on $d_1$: $[i,i,1, \dots 1]$ and $[1, \dots ,1]$,
but the second one is already fixed as a corner.
Therefore, the point $p_1$ is fixed.
A point $p_2 = [i, n-i+1,1, \dots ,1] \in d_2$ is fixed by a similar argument.

Next we show that $S$ fixes a point $p_3 = [n-i+1,i,1, \dots ,1]$.
A point $S(p_3)$ must be on $d_2$ and it must be collinear with $p_1$.
If $n$ is even there are two points on $d_2$ collinear with $p_1$: $p_2$ and $p_3$,
but $p_2$ is fixed due to step 1.
If $n$ is odd, there are 3 points collinear with $p_1$: $p_2$, $p_3$ and the face center $c_1=[\gamma, \gamma,1,\dots,1]$.
However, the point $c_1$ is fixed due to Observation~\ref{obs:FixedIntersection} because it is an intersection of the lines $d_1$ and $d_2$.
Therefore, the point $p_3$ cannot be mapped onto $c_1$.
A point $p_4 = [n-i+1, n-i+1,1, \dots ,1]$ is fixed by a similar argument.

Let $\ell_1$ be a line such that $\type(\ell_1) = (n-i+1,+,1,\dots,1)$ and $\ell_2$ be a line such that $\type(\ell_2) = (+,1,\dots,1)$.
Both lines $\ell_1$ and $\ell_2$ are fixed because $p_4,p_3 \in \ell_1$ and $\ell_2$ connects two fixed corners.
Therefore, the point $p_5$, which is an intersection of $\ell_1$ and $\ell_2$, is fixed by Observation~\ref{obs:FixedIntersection} as well.
\qed\end{proof}

In the proofs of the following lemmas we use the notions of blocks.
Let $p = [x_1, \dots, x_d]$.
We call a set $B_j(p) = \bigl\{i \in [d] \bigl| x_i = j \vee x_i = n - j + 1\bigr\}$ $j$-\emph{block} of $p$.
Note that $j$-block and $(n-j+1)$-block are the same set.
We say a line $\ell$ such that $p \in \ell$ is \emph{active} in the $j$-block $B_j(p)$ if there exists some $i \in B_j(p)$ such that $\type(\ell)_i \in \{+,-\}$.
For example $p = [1,1,2,4]$ be a point of the cube $4^4$ then $1$-block of $p$ is the set $\{1,2,4\}$ and a line $\ell \in \mathbb{L}(4^4)$ of a type $\type(\ell) = (+,1,2,-)$ is active in $B_1(p)$.
We consider only non-empty blocks.
We say that the point $p$ from the example has blocks $B_1(p)$ and $B_2(p) = \{3\}$.

\begin{lemma}
\label{lem:UniqueBlock}
Let $p$ be a point of $n^d$ and $\ell$ be a line such that $p \in \ell$.
Then, there is exactly one $j \in [n]$ such that $\ell$ is active in $B_j(p)$.
\end{lemma}
\begin{proof}
It is clear that there is at least one $j \in [n]$ such that $\ell$ is active in $B_j(p)$.
Suppose $\ell$ is active in $B_i(p)$ and $B_j(p)$, $i \neq j$.
Therefore, $p$ has some coordinates equal to $i$ or $n - i + 1$ and some coordinates equal to $j$ or $n - j + 1$.
Suppose $p$ has some coordinates equal to $i$ and $j$ (other cases are analogous).
Without loss of generality $p = [i,j,\dots]$.
Since $i \neq j$, $\type(\ell) \neq (+,+,\dots)$.
Thus, $\type(\ell) = (+,-,\dots)$.
However, it means that $j = n - i + 1$ and $B_i(p) = B_{n-i+1}(p)$.
\qed\end{proof}

\begin{lemma}
\label{lem:NonCentralDegree}
Let $p = [x_1,\dots,x_d]$ be a point of $n^d$ and it has a block $B_j(p)$ for $j \neq \gamma$.
Then, there are $2^k - 1$ active lines in $B_j(p)$ where $k = \bigl|B_j(p)\bigr|$.
\end{lemma}
\begin{proof}
For every $J \subseteq B_j(p), J \neq \emptyset$ we define a line $\ell_J$ active in $B_j(p)$ in the following way.
Let $q_J$ be a linear sequence such that $p \in q_J$ and for $i \in [d]$,
\[
\type(q_J)_i =
\begin{cases}
x_i & i \not \in J,\\
+ & i \in J \text{ and } x_i = j,\\
- & i \in J \text{ and } x_i = n-j+1.
\end{cases}
\]
For example,
\[
p = [\overbrace{n-j+1,j,\dots,j}^k,x_{k+1},\dots,x_d] \text{ and } J = \{1, 2\}
\]
the linear sequence $q_J$ has the type
\[
\type(q_J) = (-,+,j\dots,j,x_{k+1},\dots,x_d).
\]
Note that for each $J_1,J_2 \subseteq B_j(p), J_1 \neq J_2$ the linear sequences $q_{J_1}$ and $q_{J_2}$ represent different lines.

On the other hand, every line $\ell$ active in $B_j(p)$ defines a non-empty subset of $B_j(p)$ as coordinates where $\ell$ has non-constant coordinate sequences.
Therefore, the number of lines active in $B_j(p)$ is the number of non-empty subsets of $B_j(p)$, which is $2^k - 1$.
\qed\end{proof}

\begin{lemma}
\label{lem:CentralDegree}
Let $p = [x_1,\dots,x_d]$ be a point of $n^d$ and it has block $B_\gamma(p)$.
Then, there are $\frac{3^k - 1}{2}$ active lines in $B_\gamma(p)$ where $k = \bigl|B_\gamma(p)\bigr|$.
\end{lemma}
\begin{proof}
The proof is very similar to the previous one.
For every $J \subseteq B_\gamma(p), J \neq \emptyset$ we define a line $\ell_J$ such that $p \in \ell_J$.
Let $q_J$ be a linear sequence such that for $i \in \{1,\dots,d\}$
\[
\type(q_J)_i =
\begin{cases}
x_i & i \not \in J,\\
+ & i \in J.
\end{cases}
\]

And for every $K \subseteq J$ we define a linear sequence $q'_{JK}$ such that $p \in q'_{JK}$, and for $j \in \{1, \dots, d\}$,
\[
\type(q'_{JK})_j =
\begin{cases}
\type(q_J)_j & j \not \in K,\\
- & j \in K.
\end{cases}
\]
For example,
\[
p = [\overbrace{\gamma,\dots,\gamma}^k,x_{k+1},\dots,x_d] \text{ and } J = \{1, 2, 3\}, K = \{3\}
\]
the linear sequence $q'_{JK}$ has a type
\[
\type(q'_{JK}) = (+,+,-,\gamma,\dots,\gamma,x_{k+1},\dots,x_d).
\]

Note that for fixed $J,K$ and $M = J \setminus K$ the linear sequences $q'_{JK}$ and $q'_{JM}$ represent the same lines.
Again every line $\ell$ active in $B_\gamma(p)$ and its two orderings $q_1$ and $q_2$ into linear sequence define two pairs of the set $J,K \subseteq B_\gamma(p)$:
\begin{enumerate}
 \item $J = \bigl\{i \in B_\gamma(p) | \type(q_1)_i \in \{+,-\}\bigr\} =  \bigl\{i \in B_\gamma(p) | \type(q_2)_i \in \{+,-\}\bigr\}$.
 \item $K = \{i \in B_\gamma(p) | \type(q_1)_i = +\} = \{i \in B_\gamma(p) | \type(q_2)_i = - \}.$
 \item $M = \{i \in B_\gamma(p) | \type(q_1)_i = -\} = \{i \in B_\gamma(p) | \type(q_2)_i = + \}.$
\end{enumerate}

Therefore, the numbers of lines active in $B_\gamma(p)$ is a half of the number of pairs $(J,K)$ such that $J$ is a non-empty subset of $B_\gamma(p)$ and $K$ is a subset of $J$.
We have $\sum_{m=1}^{k} {\binom{k}{m}}$ choices for the set $J$.
For fixed $J$ of size $m$, we have $2^m$ choices for $K \subseteq J$.
Therefore the number of these lines is
\[
\frac{1}{2}\sum_{m=1}^{k} \Bigl({\binom{k}{m}} 2^m\Bigr) = \frac{3^{k} - 1}{2}.
\]
\qed\end{proof}

\begin{lemma}
\label{lem:Degree}
Let $n$ be odd and $\ell \in \L(n^d)$ such that the cube center $c$ is in $\ell$ and $\dim(\ell) = k$.
Let $p \in \ell, p \neq c$.
Then, $deg(p) = 2^k - 1 + \frac{3^{d-k}-1}{2}$.
\end{lemma}
\begin{proof}
Since $p \in \ell$ and $p \neq c$, the point $p$ has exactly 2 blocks $B_j(p)$ and $B_\gamma(p)$.
Note that $\bigl|B_j(p)\bigr| = k$.
Thus, the point $p$ is incident with $2^k - 1$ lines active in $B_j(p)$ and with $\frac{3^{d-k}-1}{2}$ lines active in $B_\gamma(p)$ (by Lemma~\ref{lem:NonCentralDegree} and Lemma~\ref{lem:CentralDegree}).
By Lemma~\ref{lem:UniqueBlock}, the lines active in $B_j(p)$ are disjoint from the lines active in $B_\gamma(p)$ and there are no other lines incident with $p$.
\qed\end{proof}

\begin{lemma}
\label{prp:MainDiagonal}
Every automorphism $S \in \T_n^d$ maps a main diagonal $m \in \L_m(n^d)$ onto a main diagonal $m' \in \L_m(n^d)$.
\end{lemma}
\begin{proof}
Every point on a main diagonal has only one block.
For $n$ even the proof is trivial.
For every point $q \in [n]^d$ it holds that any of blocks of $q$ is not the $\gamma$-block.
Therefore, every point $p \in m$ has degree $2^d - 1$ and any point which is not in any main diagonal has at least two blocks and thus the degree at most $2^{d - 1}$ (by Lemma~\ref{lem:NonCentralDegree}).
Every automorphism $S \in \T_n^d$ has to preserve the point degree.
Thus, a point $p \in m$ has to be mapped onto a point $p' \in m'$.

Now we prove the lemma for $n$ odd.
The center of the cube $c$ is always mapped onto $c$ ($c$ is the only point of degree $\frac{3^d-1}{2}$).
Therefore, the main diagonal $m \in \L_m(n^d)$ has to be mapped onto a line $\ell \in \L(n^d)$ such that $c \in \ell$.
By Lemma~\ref{lem:Degree}, we know the degree of a non-central point $p \in \ell$ is $\deg(p)=2^k - 1 + \frac{3^{d-k}-1}{2}$.

The degree of a non-central point $q \neq c$ on a main diagonal $m \in \L_m(n^d)$ is $\deg(q) = 2^d - 1$.
We show that if $k \neq d$ then $2^k - 1 + \frac{3^{d-k}-1}{2} \neq 2^d - 1$.
For contradiction let us suppose that $2^d - 2^k=  \frac{3^{d-k}-1}{2}$ and $k < d$.
We rewrite the formula into binary numbers:
\[
\begin{array}{rrl}
2^d & 1 \overbrace{0 \dots\dots\dots 0}^{d} &\\
-2^k & -1 \overbrace{0\dots 0}^{k} &\\
\hline
\frac{3^{d-k}-1}{2} & \overbrace{1\dots 1 \underbrace{0 \dots 0}_k}^d  &=\beta > 0.\\
\end{array}
\]

It is easy to prove by induction that $4$ divides $3^{d-k} - 1$ if and only if $d - k$ is even.
The number $\beta$ must be even so $d - k$ must be even as well.
We use the well-known divisibility-by-3 test in the binary system
for $\delta = 2\beta + 1$ (it should be equal to $3^{d-k} > 1$).
The binary number is divisible by 3 if and only if the number $E$ of even order digits
and the number $O$ of odd order digits are equal modulo 3.
Note that
\[
\delta = \overbrace{1 \dots 1 }^{d-k}\overbrace{0 \dots 0}^{k}1.
\]
The number $d-k$ is even, thus the numbers of digits of the orders $1$ to $d$ are equal,
but $|E-O|=1$ (because of the $1$ at the order 0).
Therefore $\delta$ is not divisible by 3, which is the contradiction.
\qed\end{proof}

\begin{lemma}
\label{prp:LineDim1}
Let $S \in \T_n^d$, $e$ be an edge and $p$ be a corner, such that $p \in e$.
If the corner $p$ is fixed by $S$, then $S(e) = e'$ is an edge such that $p \in e'$.
\end{lemma}

\begin{proof}
Without loss of generality the corner $p$ is $[1,\dots,1]$ and the type of $e$ is $\type(e) = (+,1,\dots,1)$.
First we prove the lemma for odd $n$.
Let $k = \dim\bigl(S(e)\bigr)$ and suppose that $1 < k < d$ (the line $S(e)$ can not have a dimension $d$ as main diagonals are mapped on to main diagonals by Lemma~\ref{prp:MainDiagonal}).
Without loss of generality 
\[
\type\bigl(S(e)\bigr) = (\overbrace{+,\dots,+}^k,1,\dots,1).
\]
Let $c_1$ be the center of $e$, i.e., the point $[\gamma,1,\dots,1]$.
Note that $c_1$ is collinear with the cube center $c$.
Thus, the point $S(c_1)$ has to be also collinear with the cube center and
\[
 S(c_1) = [\overbrace{\gamma,\dots,\gamma}^k,1,\dots,1].
\]
Consider the set of lines
\[
 L = \bigl\{\ell \in \L(n^d) \mid \forall i \leq k: \type(\ell)_i \in \{+,-\}, \forall i > k: \type(\ell)_i = 1 \bigr\}.
\]
Note that the set $L$ contains all lines incident to the vertex $S(c_1)$ which are active in the block $B_\gamma\bigl(S(c_1)\bigr)$.
Moreover, each line in $L$ intersect exactly 2 main diagonals and the intersections points are corners (in particular not the cube center $c$).
The line $S(e)$ is in $L$.
Since $k > 1$, there is a line $\ell' \in L$ different from $S(e)$.
Let $\ell$ be a preimage of $\ell'$, i.e., $\ell = S(\ell')$.
The line $\ell$ has to intersect exactly 2 main diagonals as main diagonals are mapped onto main diagonals by Lemma~\ref{prp:MainDiagonal}.
Moreover, the line $\ell$ can not intersect the main diagonals in the cube center $c$.
Note that $\ell \neq e$.
There is only one line $\ell_1$ incident to $c_1$, different from $e$, such that it intersects some main diagonal.
The type of $\ell_1$ is
\[
 \type(\ell_1) = (\gamma,+,\dots,+).
\]
However, the line $\ell_1$ intersects the main diagonals in the cube center $c$.
Thus, the line $\ell'$ does not have a preimage, which is a contradiction and $k = \dim\bigl(S(e)\bigr)= 1$.

We now complete the proof for even $n$.
For a contradiction suppose that $\dim(e') \geq 2$.
Without loss of generality the type of $e'$ is $(+,\dots,+,1,\dots,1)$.
Let $p_2 = [2,1,\dots,1]$ and $p_3 = [3,1,\dots,1]$.
Since $n \geq 4$, the points $p_2$ and $p_3$ are not corners.
Therefore, the point $p_i$ (for $i \in \{2,3\}$) has blocks $B_i(p_i) = \{1\}$ and $B_1(p_i) = \{2,\dots,d\}$.
Let $L_2$ be a set of lines incident with $p_2$ without the edge $e$ and similarly $L_3$ be a set of lines incident with $p_3$ without $e$.
Note that lines in $L_i$ (for $i \in \{2,3\}$) can be active only in the block $B_1(p_i)$.
Let $\ell_2 \in L_2$ and $\ell_3 \in L_3$.
For $\ell_2$ holds that $\type(\ell_2)_1 = 2$ and for $\ell_3$ holds that $\type(\ell_3)_1 = 3$.
Therefore, the lines $\ell_2$ and $\ell_3$ cannot intersect.

Now take images of $p_2$ and $p_3$.
Let $q_2 = S(p_2) = [i,\dots,i,1,\dots,1]$ and $q_3 = S(p_3) = [j,\dots,j,1,\dots,1]$.
Since $\dim(e') \geq 2$, the point $q_2$ is incident with a line $k_2$ such that $\type(k_2) = (+,i,\dots,i,1,\dots,1)$.
Similarly, the point $q_3$ is incident with a line $k_3$ such that $\type(k_3) = (j,+,\dots,+,1,\dots,1)$.
The lines $k_1$ and $k_2$ have to be images of some lines in $L_1$ and $L_2$, respectively.
However, the lines $k_1$ and $k_2$ intersect in a point $[j,i,\dots,i,1,\dots,1]$, which is a contradiction.
\qed\end{proof}

\begin{lemma}
\label{lem:FixingCorners}
If an automorphism $S \in \T_n^d$ fixes the corner $[1,\dots,1]$ and all its neighbors, then $S$ fixes all corners of the cube $n^d$.
\end{lemma}
\begin{proof}
We prove the automorphism $S$ fixes all corners $p = [x_1, \dots ,x_d]$
by induction over $$k(p)  = \bigl|\{i \in [d]:x_i = n\}\bigr|.$$
By the assumption, the automorphism $S$ fixes corners $p$ such that $k(p) \in \{0,1\}$.
Without loss of generality, a corner $q$ such that $k(q) > 1$ has coordinates
\[
	q = [\overbrace{n , \dots ,n}^{k(q)}, 1, \dots ,1].
\]
We take neighbors $q_1, q_2$ of the corner $q$ as
\begin{align*}
q_1 &= [\overbrace{n, \dots ,n}^{k(q)-1}, 1, \dots ,1] \\
q_2 &= [1, \underbrace{n, \dots ,n}_{k(q)-1}, 1, \dots ,1].
\end{align*}
The corners $q_1$ and $q_2$ have two common neighbors: $q$ and
\[
q_3 = [1, \overbrace{n, \dots ,n}^{k(q)-2}, 1, \dots, 1].
\]
Corners $q_1$, $q_2$ and $q_3$ are fixed by the induction hypothesis.
Therefore, corner $q$ is also fixed as it must be the neighbor of $q_1$ and $q_2$.
\qed\end{proof}

\section{Generators of the Group $\T_n^d$}
\label{sec:GeneralCube}

In this section we characterize the generators of the group $\T_n^d$.
As we stated in Section~\ref{sec:Intro}, we use the groups $\G_n^d$, $\F_{n}$.

\begin{definition}
Let $\mathbb{A}_n^d$ be a group generated by elements of $\G_n^d \cup \F_n$.
\end{definition}
We prove that $\mathbb{A}_n^d = \T_n^d$.
The idea of the proof, that resembles a similar proof of Silver~\cite{silver67}, is composed of two steps:
\begin{enumerate}
\item For any automorphism $S \in \T_n^d$ we find an automorphism $A \in \mathbb{A}_n^d$,
such that $S \circ A$ fixes all corners of the cube $n^d$ and one edge.
\item If an automorphism $S' \in \T_n^d$ fixes all corners and one edge then $S'$ is the identity.
\end{enumerate}
Hence, for every $S \in \T_n^d$ we find an inverse element $A$
such that $A$ is composed only by elements of $\A^d_n$,
therefore $S \in \mathbb{A}_n^d$.
We divide the construction of the automorphism $A$ into two steps.
In the proof of Theorem~\ref{thm:FixCorners} we construct an automorphism $A' \in \A_n^d$ such that $S \circ A'$ fixes all corners of the cube.
In the proof of Theorem~\ref{thm:FixEdge} we construct an automorphism $A'' \in \A_n^d$ such that $S \circ A' \circ A''$ fixes all corners and one edge of the cube.

In the next proofs we use the following permutations.
For $i,j \in [n]$ and $i,j \neq \gamma$ (in a case of odd $n$), let $\rho = \rho(i,j)$ be a permutation in $\S_n$ such that
\begin{enumerate}
\item $\rho(i) = j, \rho(j) = i$.
\item $\rho(n - i + 1) = n - j + 1, \rho(n - j + 1) = n - i + 1$.
\item $\rho(k) = k$ for all other $k \not \in \{i,j,n-i+1,n-j+1\}$.
\end{enumerate}
Note that the permutation $\rho(i,j)$ has exactly two cycles of the length two (or one cycle if $i = n - j + 1$) and it has the symmetry property, i.e. $F_{\rho(i,j)} \in \F_n$.
\begin{theorem}
\label{thm:FixCorners}
For all $S \in \T_n^d$ there exists $A' \in \mathbb{A}_n^d$ such that $S \circ A'$ fixes every corner of the cube $n^d$.
\end{theorem}
\begin{proof}
We start with the point $p_0 = [1, \dots ,1]$.
By Lemma~\ref{prp:MainDiagonal}, the point $S(p_0) = [x_1,\dots,x_d]$ has to be on a main diagonal, i.e., there is some $j \in [n]$ such that each $x_i$ is equal $j$ or $n - j + 1$.
The point $p_0$ cannot be mapped onto the cube center, thus $j \neq \gamma$.
We choose $F = F_{\rho(j, 1)} \in \mathbb{F}_n$.
Thus, $S \circ F(p_0)$ is a corner.
Then, we choose a translation $T_a \in \G_n^d$ where $a_i = 1$ if and only if $\bigl[S \circ F(p_0)\bigr]_i = n$.
Therefore, $S \circ F \circ T_a (p_0) = p_0$.

By induction over $i$ we can construct automorphisms $Z_i$ to fix the points $p_0$ and
\[
p_i = [1, \dots ,\underset{i}{n}, \dots, 1]
\]
for all $i \in \{1, \dots ,d\}$.
For $i = 0$, the point $p_0$ is fixed by $Z_0 = S \circ F \circ T_a$.
For $i > 0$, by induction hypothesis we have an automorphism $Z_{i-1}$ such that it fixes all points in the set $P_{i-1} = \bigl\{p_k|0 \leq k \leq i -1\bigr\}$.
The corner $p_i$ is mapped onto $p_j$ for $j \geq i$ because edges incident with $p_0$ are mapped onto edges incident with $p_0$ (by Lemma~\ref{prp:LineDim1})
and points in $P_{i-1}$ are already fixed.
If $Z_{i-1}(p_i) = p_i$, we choose $Z_i = Z_{i-1}$.
Otherwise, we choose a rotation $R_\pi$ where $\pi$ switches $i$ and $j$ coordinates, thus
\begin{align*}
 &R_\pi\bigl([x_1, \dots ,x_i, \dots ,x_j, \dots ,x_d]\bigr) = [x_1, \dots ,x_j, \dots ,x_i, \dots ,x_d]\colon \\
&R_\pi(p_j) = R_\pi\bigl([1, \dots ,\underset{i}{1,} \dots ,\underset{j}{n}, \dots ,1]\bigr)
= [1, \dots ,\underset{i}{n}, \dots ,\underset{j}{1}, \dots ,1] 
\end{align*}
We set $Z_i = Z_{i-1} \circ R_\pi$.
Hence, the automorphism $Z_i$ fixes $p_i$ and all points of $P_i$ because the rotation $R_\pi$ does not affect the first $i - 1$ coordinates.
Note that it also fixes $p_0$.
In this way we can fix all corners $p_i$ for $i \in \{0, \dots ,d - 1\}$.
Thus, the automorphism $Z_{d-1}$ fixes all points of $P_{d-1}$ and the corner $p_d$ is fixed automatically
because there is no other possibility where the corner $p_d$ can be mapped.
The automorphism $Z_{d-1}$ fixes the corner $p_0 = [1,\dots,1]$ and all its neighbors.
Therefore by Lemma~\ref{lem:FixingCorners}, the automorphism $Z_{d-1} = S \circ A'$ for some $A' \in \A^n_d$ fixes all corners of the cube.
\qed\end{proof}

\begin{theorem}
\label{thm:FixEdge}
For all $S \in \T_n^d$ there exists $A \in \mathbb{A}_n^d$ such that $S \circ A$ fixes every corner of the cube $n^d$ and every point of a line $\ell = \bigl\{[i, 1, \dots ,1] | i \in [n]\bigr\}$.
\end{theorem}
\begin{proof}
By Theorem~\ref{thm:FixCorners} we have an automorphism $A' \in \A_n^d$ such that $S' = S \circ A'$ fixes all corners of the cube.
We find an automorphism $A'' \in A_n^d$ such that $S' \circ A''$ fixes all corners and all points on the line $\ell$.
The line $\ell$ is fixed by $S'$ due to Observation~\ref{obs:FixedLine}.
Let $k = \lfloor \frac{n}{2} \rfloor$.
We construct the automorphism $A''$ by induction over $i \in \{1, \dots, k\}$.
We show that in a step $i$ an automorphism $Y_i$ fixes all corners and every point in the set
\[
Q_i = \bigl\{[j,1,\dots,1], [n-j+1,1,\dots,1]|1\leq j \leq i\bigr\}.
\]

First, let $i = 1$ and $Y_1 = S'$.
The automorphism $Y_1$ fixes all corners and $Q_1$ contains only $[1,\dots,1]$ and $[n,1,\dots,1]$, which are also corners.
Suppose that $i > 1$.
By induction hypothesis, we have an automorphism $Y_{i-1}$ which fixes all corners and every point in the set $Q_{i-1}$.
If $Y_{i-1}\bigl([i,1,\dots,1]\bigr) = [i,1,\dots,1]$ then $Y_i = Y_{i-1}$.
Otherwise, $Y_{i-1}\bigl([i,1,\dots,1]\bigr) = [j,1,\dots,1]$.
Note that $i < j \leq n-i+1$ because points in $Q_{i-1}$ are already fixed.
Also in the case of odd $n$, it holds that $j \neq \gamma$ as $i \neq \gamma$ and $[i,1,\dots,1]$ is not collinear with the cube center, thus the point $Y_{i-1}\bigl([i,1,\dots,1]\bigr)$ is not collinear with the cube center as well.
Let us consider $F_\rho \in \mathbb{F}_n$ for $\rho = \rho(i,j)$.
The automorphism $Y_{i} = Y_{i-1} \circ F_\rho$ fixes the following points:
\begin{enumerate}
\item All corners, as the automorphism $Y_{i-1}$ fixes all corners
by the induction hypothesis and $\rho(1) = 1$ and $\rho(n)=n$.
\item Set $Q_{i-1}$, as the automorphism $Y_{i-1}$ fixes the set $Q_{i-1}$
by the induction hypothesis and $\rho(s) = s$ for all $s < i$ and $s > n - i + 1$.
\item Point $[i,1,\dots,1]$: $Y_{i-1} \circ F_\rho \bigr([i,1,\dots,1]\bigl) = F_\rho \bigr([j,1,\dots,1]\bigl) = [i,1,\dots,1]$.
\item Point $[n-i+1,1,\dots,1]$ by Lemma~\ref{lem:FrontLine}.
\end{enumerate}
Note that if $n$ is odd the point $[\gamma,1,\dots,1]$ is fixed as well by an automorphism $Y_k$.
Thus, the automorphism $Y_k = S \circ A$ for some $A \in \A_n^d$ fixes all points of the line $\ell$ and all corners of the cube.
\qed\end{proof}

It remains to prove that if an automorphism $S \in \T_n^d$ fixes all corners and all points in the line $\ell = \bigr\{[i, 1,\dots, 1]| i \in [n]\bigl\}$ then $S$ is the identity.
We prove it in two parts.
First, we prove that if $d = 2$ then the automorphism $S$ is the identity.
Then, we prove it for a general dimension by an induction argument.

\begin{theorem}
\label{thm:FixFrontFace}
Let an automorphism $S \in \T_n^2$ fixes all corners of the cube $n^2$ and all points in the line $\ell = \bigr\{[i, 1]| i \in [n]\bigl\}$.
Then, the automorphism $S$ is the identity.
\end{theorem}

\begin{proof}
Let $d_1, d_2 \in \L_m(n^2)$.
Thus, $\type(d_1) = (+,+)$ and $\type(d_2)= (+,-)$.
Since all corners are fixed, the diagonals $d_1$ and $d_2$ are fixed as well due to Observation~\ref{obs:FixedLine}.
Let $p \in d_1 \cup d_2$ such that $p$ is not a corner.
The point $p$ is collinear with the only one point $q \in \ell$ such that $q$ is not a corner.
Therefore, every point on the diagonals $d_1$ and $d_2$ is fixed.

Now we prove that every line in $\L(n^2)$ is fixed.
Let $\ell_1 \in \L(n^2)$ be a line of a dimension 1.
Suppose $n$ is even.
The line $\ell_1$ intersects the diagonals $d_1$ and $d_2$ in distinct points, which are fixed.
Therefore, the line $\ell_1$ is fixed as well by Observation~\ref{obs:FixedLine}.

Now suppose $n$ is odd.
If $\ell_1$ does not contain the cube center $c_1 = [\gamma, \gamma]$ then $\ell$ is fixed by the same argument as in the previous case.
Thus, suppose $c_1 \in \ell_1$.
There are two lines $\ell_2, \ell_3$ in $\L(n^2)$ of dimension 1 which contains $c_1$. 
Their types are $\type(\ell_2) = (\gamma, +)$ and $\type(\ell_3) = (+,\gamma)$.
The line $\ell_2$ also intersects the line $\ell$.
Therefore, the line $\ell_2$ contains two fixed points $c_1$ and $[\gamma, 1]$ and thus it is fixed.
The line $\ell_3$ is fixed as well because every other line is fixed.
For better understanding of all lines and points used in the proof see Figure~\ref{fig:fixAllFrontFace} with an example of the cube $5^2$.

\begin{figure}[ht]
 \centering
 \includegraphics{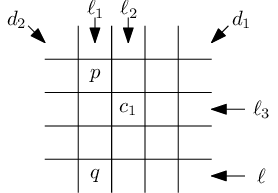}
 \caption{Points and lines used in the proof of Theorem~\ref{thm:FixFrontFace}.}
 \label{fig:fixAllFrontFace}
\end{figure}

Every point in $n^2$ is fixed due to Observation~\ref{obs:FixedIntersection}
because every point is in an intersection of at least two fixed lines.
\qed\end{proof}

\begin{theorem}
\label{thm:IdentityGeneral}
Let an automorphism $S \in \T_n^d$ fix all corners of the cube $n^d$ and all points of an arbitrary edge $e$.
Then, the automorphism $S$ is the identity.
\end{theorem}
\begin{proof}
We prove the theorem by induction over dimension $d$ of the cube $n^d$.
The basic case for $d=2$ is Theorem~\ref{thm:FixFrontFace}.

Therefore, we can suppose $d > 2$ and the theorem holds for all dimensions smaller then $d$.
Without loss of generality, $e = \bigl\{[i,1,\dots,1]\bigl|i \in [n]\bigr\}$.
For $s \in \{1,n\}$ and $i \in [d]$, let $F^s_i$ be a $(d-1)$-dimensional face which fix the $i$-th coordinate to $s$, i.e.,
\[
 F^s_i = \bigl\{[x_1,\dots,x_d] | x_i = s, x_j \in [n] \text{ for } j \neq i\bigr\}.
\]
Note that the faces $F^1_2,\dots,F^1_d$ cotanins the line $e$.
Therefore, all points of the faces $F^1_2,\dots,F^1_d$ are fixed by the induction hypothesis as all corners are fixed as well.
Let $e_1$ be an edge of type $(1,1,+,1,\dots,1)$ (since $d > 2$, the edge $e_1$ is well-defined).
It holds that $e_1 \subseteq F^1_1 \cap F^1_2$.
Thus, all corners of the face $F^1_1$ are fixed and all points of one edge of $F^1_1$ are fixed.
Therefore, all points of face $F^1_1$ are fixed by the induction hypothesis.

We will prove that all points of faces $F^n_i$ are fixed by similar argument.
Let $f_i$ be an edge of type 
\[
 \type (f_i) = 
 \begin{cases}
 (1,\dots,1,+,\underset{i}{n},1,\dots,1) & \text{if $i > 1$} \\
 (n,1,\dots,1,+) & \text{if $i = 1$}.
 \end{cases}
\]
Consider the face $F^n_i$ and set $j = i + 1 \mod d$.
The line $f_i$ is contained in the faces $F^n_i$ and $F^1_j$.
Thus by the same argument as above, all points of the faces $F^n_1,\dots,F^n_d$ are fixed.

We showed that every outer point is fixed.
Every line $\ell \in \L(n^d)$ is fixed due to Observation~\ref{obs:FixedLine} because every line contains at least two outer points.
Therefore by Observation~\ref{obs:FixedIntersection}, every point $q \in [n]^d$ is fixed because every point is an intersection of at least two lines.
\qed\end{proof}

\begin{corollary}
The groups $\T_{2k}^d$ and $\T_{2k+1}^d$ are isomorphic for $k \geq 2$.
\end{corollary}
\begin{proof}
The groups $\G_n^d$ are all isomorphic to the group of $d$-dimensional hypercube automorphism for all $n$.
For every permutation $\pi \in \mathbb{S}_{2k+1}$ with the symmetry property holds that $\pi(k) = k$.
Therefore, the group $\F_{2k}$ is isomorphic to the group $\F_{2k+1}$.
\qed\end{proof}

\section{Order of the Group $\T_n^d$}
\label{sec:SizeTnd}

In the previous section we characterized the generators of the group $\T_n^d$.
Now we compute the order of $\T_n^d$.
First, we state several technical lemmas.

\begin{lemma}
\label{lem:BasicOrder}
Orders of the basic groups are as follows.
\begin{enumerate}
\item $|\G_n^d|=2^{d}d!$.
\item $|\F_n|=2^kk!$ for $k = \lfloor\frac{n}{2}\rfloor$.
\end{enumerate}
\end{lemma}

\begin{proof}
Size of the hypercube automorphism group $\Q_d$ is well known~\cite{godsil04}.

Size of the group $\F_n$ is the number of permutations $\pi \in \S_n$ with the symmetry property.
If $n$ is even, we have $n$ possibilities how to choose the image of the first element,
we have $n-2$ possibilities for the second element, and so on, thus there are
\[
 \prod^{\frac{n}{2} - 1}_{i=0} (n-2i)
\]
such permutations.
If $n$ is odd, the element $\frac{n+1}{2}$ has to be mapped onto itself.
Therefore, the order of $\mathbb{F}_n$ where $n$ is odd is the same as the order of $\mathbb{F}_{n-1}$.
The general formula is
\[
 \prod_{i=0}^{k-1}{(2k - 2i)}= 2^k \prod_{i=0}^{k-1}{(k - i)} = 2^kk!
\]
for $k = \lfloor\frac{n}{2}\rfloor$.
\qed\end{proof}

\begin{lemma}
\label{lem:Commuting}
The groups $\G_n^d$ and $\F_n$ commute.
\end{lemma}
\begin{proof}
Let $T_a \circ R_\pi \in \G_n^d$ and $F_\rho \in \F_n$. 
Note that for $\rho$ holds that $\rho\bigl(\flip(i,b)\bigr) = \flip\bigl(\rho(i),b\bigr)$ due to the symmetry property.
Then,
\begin{align*}
T_a \circ R_\pi \circ F_\rho & \bigl([x_1, \dots ,x_d]\bigr) \\
=& F_\rho \Bigl(\bigl[\flip(x_{\pi(1)},a_{\pi(1)}),\dots,\flip(x_{\pi(d)},a_{\pi(d)})\bigr]\Bigr) \\
=& \Bigl[\rho\bigl(\flip(x_{\pi(1)},a_{\pi(1)})\bigr), \dots ,\rho\bigl(\flip(x_{\pi(d)},a_{\pi(d)})\bigr)\Bigr] \\
=& \Bigl[\flip\bigl(\rho(x_{\pi(1)}),a_{\pi(1)}\bigr), \dots ,\flip\bigl(\rho(x_{\pi(d)}),a_{\pi(d)}\bigr)\Bigr].
\end{align*}

Similarly,
\begin{align*}
F_\rho \circ T_a \circ R_\pi & \bigl([x_1, \dots , x_d]\bigr) = T_a \circ R_\pi \Bigl(\bigl[\rho(x_1), \dots, \rho(x_d)\bigr]\Bigr) \\
=& \Bigl[\flip\bigl(\rho(x_{\pi(1)}),a_{\pi(1)}\bigr), \dots ,\flip\bigl(\rho(x_{\pi(d)}),a_{\pi(d)}\bigr)\Bigr].
\end{align*}
\qed\end{proof}


By Lemma~\ref{lem:Commuting} we can conclude that any automorphism $A \in \T_n^d$ can be written as $A = G \circ F$ where $G \in \G_n^d$ and $F \in \F_n$.
Thus, the product
\[
\G_n^d \F_n = \bigl\{ G \circ F | G \in \G_n^d, F \in \F_n \bigr\}
\]
is exactly the group $\T_n^d$.
We state the well-known product formula for a group product.
\begin{lemma}[Product formula~\cite{ballaster10}]
\label{lem:ProductFormula}
Let $S$ and $T$ be subgroups of a finite group $G$.
Then, for an order of a product $ST$ holds that
\[
|ST| = \frac{|S|\cdot|T|}{|S \cap T|}.
\]
\end{lemma}
Thus, for computing the order of $\T_n^d$ we need to compute the order of the intersection of the basic groups $\G_n^d$ and $\F_n$.

\begin{lemma}
\label{lem:RintersectionF}
The intersection $\G_n^d \cap \F_n = \{\Id, F_\sigma\}$ where $\sigma(i) = n - i + 1$.
\end{lemma}
\begin{proof}
It is clear that $\{\Id, F_\sigma\} \subseteq \G_n^d \cap \F_n$ because $F_\sigma = T_a$ where $a = (1,\dots,1)$.

Consider a main diagonal $\ell = \bigl\{ [i, \dots ,i] | i \in [n]\bigr\}$ and its ordering into a linear sequence $(p^1,\dots,p^n)$, where $p^i = [i,\dots,i]$.
Every automorphism $G \in \G_n^d$ preserves an order of points on the line $\ell$, i.e., a sequence $\bigl(G(p^1),\dots,G(p^n)\bigr)$ is an ordering of $G(\ell)$ into a linear sequence.

Consider an automorphism $F_\rho \in \F_n$.
We claim that $\bigl(F_\rho(p^1),\dots,F_\rho(p^n)\bigr)$ is an ordering of $F_\rho(\ell)$ into a linear sequence if and only if $\rho$ is the identity or $\sigma$.
Recall that $F_\rho(p^i) = \bigl[\rho(p^i_1),\dots,\rho(p^i_d)\bigr]$.
Thus, to $\bigl(F_\rho(p^1),\dots,F_\rho(p^n)\bigr)$ be a linear sequence it must hold that $\rho(i) = i$ or $\rho(i) = n - i + 1$ for all $i$.
We conclude that if $\rho$ is not the identity and $\rho \neq \sigma$ then $F_\rho \not \in \G_n^d$.
\qed\end{proof}
As a corollary of Lemmas~\ref{lem:BasicOrder},~\ref{lem:ProductFormula} and~\ref{lem:RintersectionF} we get the second part of Theorem~\ref{thm:MainResult}.

\section{The Complexity of Colored Cube Isomorphism}
\label{sec:Complexity}
In this section we prove Theorem~\ref{thm:Complexity}.
As we stated before, {\sc CHI} is in ${\GI}$.
Therefore, {\ColoredCube} as a subproblem of {\sc CHI} is in ${\GI}$ as well.
It remains to prove the problem is $\GI$-hard.
Let $s_1$ and $s_2$ be colorings of a combinatorial cube $n^d$.
We say the colorings $s_1$ and $s_2$ are \emph{isomorphic} if there is an automorphism $A \in \T^d_n$ which preserves the colors, i.e., for every point $p$ of a combinatorial cube $n^d$ holds that $s_1(p) = s_2\bigl(A(p)\bigr)$.

First, we describe how we reduce the input of {\sc Graph Isomorphism} to the input of {\ColoredCube}.
Let $G = (V, E)$ be a graph.
Without loss of generality $V = [n]$.
We construct the coloring $s^G : [k]^2 \rightarrow \{0,1\}$ for  $k = 2n + 4$ as follows.
The value of $s^G\bigl([i,j]\bigr)$ is 1 if $[i,j] = [n+1,n+1]$ or $[i,j] = [n+1,n+2]$ or $i,j \leq n$ and $\{i,j\} \in E$.
The value of $s^G(p)$ for any other point $p$ is 0.
We can view the coloring $s^G$ as a matrix $M^G$ such that $M^G_{i,j} = s^G\bigl([i,j]\bigr)$.
The submatrix of $M^G$ consisting of the first $n$ rows and $n$ columns is exactly the adjacency matrix of the graph $G$.
For example, let $P$ be a path on 3 vertices, then
\[
M^P =
\begin{pmatrix}
0 & \mathbf{1} & 0 & 0 & 0 & 0 & 0 & 0 & 0 & 0 \\
\mathbf{1} & 0 & \mathbf{1} & 0 & 0 & 0 & 0 & 0 & 0 & 0 \\
0 & \mathbf{1} & 0 & 0 & 0 & 0 & 0 & 0 & 0 & 0 \\
0 & 0 & 0 & \mathbf{1} & \mathbf{1} & 0 & 0 & 0 & 0 & 0 \\
0 & 0 & 0 & 0 & 0 & 0 & 0 & 0 & 0 & 0 \\
0 & 0 & 0 & 0 & 0 & 0 & 0 & 0 & 0 & 0 \\
0 & 0 & 0 & 0 & 0 & 0 & 0 & 0 & 0 & 0 \\
0 & 0 & 0 & 0 & 0 & 0 & 0 & 0 & 0 & 0 \\
0 & 0 & 0 & 0 & 0 & 0 & 0 & 0 & 0 & 0 \\
0 & 0 & 0 & 0 & 0 & 0 & 0 & 0 & 0 & 0 \\
\end{pmatrix}
\]

The idea of the reduction is as follows.
If two colorings $s^{G_1}, s^{G_2}$ are isomorphic via a cube automorphism $A \in \T_k^2$ then $A$ can be composed only of automorphisms in $\F_k$ (due to the colors of $[n+1,n+1]$ and $[n+1,n+2]$).
Hence, the automorphism $A = F_\rho$ for some permutation $\rho \in \S_k$.
Moreover, the permutation $\rho$ maps the numbers in $[n]$ to the numbers in $[n]$ and describes the isomorphism between the graphs $G_1$ and $G_2$.

\begin{lemma}
\label{lem:PermutationAutomorphism}
Let $G_1, G_2$ be graphs without vertices of degree 0.
If colorings $s^{G_1}$, $s^{G_2}$ are isomorphic via a cube automorphism $A \in \T_k^2$ then $A = F_\rho \in \F_k$.
Moreover, $\rho(i) \leq n$ if and only if $i \leq n$.
\end{lemma}
\begin{proof}
Let $A = S \circ F$ where $S \in \G_k^2, F \in \F_k$ and $m_1, m_2$ be main diagonals of $[k]^2$ of a type $(+,+)$ and $(+,-)$, respectively.
Due to the colors of $p_1 = [n+1,n+1]$ and $p_2 = [n+1,n+2]$ we will show that $A$ has to fix $m_1$ and $m_2$ and that $A \in \F_k$.

Since $G_1$ and $G_2$ are simple graphs without loops, there is exactly one point of the color 1 on the main diagonal $m_1$ (the point $p_1$) and there are no points of the color 1 on the main diagonal $m_2$ in both colorings $s^{G_1}$ and $s^{G_2}$.
Therefore, $A$ has to fix $m_1$ and $m_2$ and the point $p_1$.
Let $\ell_1$ be a line of a type $(n+1, +)$ and $\ell_2$ be a line of a type $(+,n+1)$.
Note that in both coloring the line $\ell_1$ contains two points of the color 1 ($p_1$ and $p_2$) and $\ell_2$ contains only one point of the color 1 (the point $p_1$).
The line $\ell_1$ can be mapped only on the lines $\ell_1$ or $\ell_2$.
However, due to the colors of the points $p_1$ and $p_2$ in both coloring the line $\ell_1$ has to be fixed.
Thus, the point $p_2$ is fixed as well.

Every automorphism in $\F_k$ fixes the lines $m_1$ and $m_2$.
Thus, the automorphism $S$ has to fix the main diagonals as well.
Let $S = T_a \circ R_\pi$.
There are 8 automorphisms in $\G_k^2$.
By simple case analysis we know that $S$ fixes the lines $m_1$ and $m_2$ if and only if $a = (0,0)$ or $a = (1,1)$.
If $\pi$ is the identity, then $S = T_a$ and it is also in $\F_k$ (see Lemma~\ref{lem:RintersectionF}).

Let us suppose that $\pi \neq \Id$ and $a = (0,0)$.
Note that $A(p_1) = \bigl[\rho(n+1),\rho(n+1)\bigr]$.
The automorphism $A$ fixes the point $p_1$, thus $\rho(n+1) = n+1$.
Therefore, $A(p_2) = \bigl[\rho(n+2),\rho(n+1)\bigr] = \bigl[\rho(n+2),n+1\bigr] \neq p_2$, which is a contradiction as we proved that $A$ fixes the point $p_2$.
The proof for $a = (1,1)$ is identical.
Thus, we conclude that $A \in \F_k$.

Now we prove the last part of the lemma.
We already know that $\rho(n+1) = n+1$ and $\rho(n+2) = n+2$.
For every $i \leq n$ there is at least one point with color~1 on a line of type $(+,i)$ in both colorings $s^{G_1}, s^{G_2}$ because graphs $G_1$ and $G_2$ do not contain any vertex of degree 0.
On the other hand, for every $i \geq n + 3$ there are only points with color 0 on a line of type $(+,i)$ in both colorings.
Therefore, if $i \leq n$ then $i$ has to be mapped on $j \leq n$ by $\rho$.
\qed\end{proof}

\noindent The proof of the following theorem follows from Lemma~\ref{lem:PermutationAutomorphism}.
\begin{theorem}
\label{thm:Reduction}
Let $G_1 = (V_1, E_1)$ and $G_2 = (V_2, E_2)$ be graphs without vertices of degree 0.
Then, the graphs $G_1$ and $G_2$ are isomorphic if and only if the colorings $s^{G_1}$ and $s^{G_2}$ are isomorphic.
\end{theorem}
\begin{proof}
First, suppose that $s^{G_1}$ and $s^{G_2}$ are isomorphic.
Let $V_1 = V_2 = [n]$.
By Lemma~\ref{lem:PermutationAutomorphism}, we know that $s^{G_1}$ and $s^{G_2}$ are isomorphic via a cube automorphism $F_\rho \in \F_k$.
We define the function $f: V_1 \rightarrow V_2$ as $f(i) = \rho(i)$.
By Lemma~\ref{lem:PermutationAutomorphism}, $f$ is a well defined bijection.
It remains to prove that $f$ is a graph isomorphism:
\[
\{i,j\} \in E_1 \Leftrightarrow s^{G_1}\bigl([i,j]\bigr) = 1 \Leftrightarrow s^{G_2}\bigl([\rho(i),\rho(j)]\bigr) = 1 \Leftrightarrow \bigl\{f(i), f(j)\bigr\} \in E_2.
\]

Now we prove the other implication.
Let $f: V(G_1) \rightarrow V(G_2)$ be an isomorphism of $G_1$ and $G_2$.
We construct the permutation $\rho: [k] \rightarrow [k]$ as follows:
\[
 \rho(i) =
 \begin{cases}
  i & n+1 \leq i  \leq n + 2 \\
  f(i) & i \leq n
 \end{cases}
\]
We define values of $\rho(i)$ for $i \geq n + 3$ in such a way the symmetry property holds for the permutation $\rho$.

We prove that $F_\rho \in \F_k$ is an isomorphism between $s^{G_1}$ and $s^{G_2}$.
Let us suppose that $s^{G_1}\bigl([i,j]\bigr) = 1$.
If $[i,j] = [n+1,n+1]$ or $[i,j] = [n+1, n+2]$ then $s^{G_2}\bigl(F_\rho([i,j])\bigr) = 1$ as well.
Otherwise, $i,j \leq n$ because there is no other point colored by 1.
Thus,
\[
 s^{G_1}\bigl([i,j]\bigr) = 1 \Leftrightarrow \{i,j\} \in E_1 \Leftrightarrow \bigl\{f(i),f(j)\bigr\} \in E_2 \Leftrightarrow s^{G_2}\bigl([\rho(i),\rho(j)]\bigr) = 1.
\]
Hence, we proved that $s^{G_1}\bigl([i,j]\bigr) = 1$ if and only if $s^{G_2}\bigl(F_\rho[i,j]\bigr) = 1$.
\qed\end{proof}


%

We may suppose that the input graphs $G_1$ and $G_2$ have minimum degree at least
1 for the purpose of the polynomial reduction of {\sc Graph Isomorphism} to \ColoredCube.
Thus, Theorem~\ref{thm:Complexity} follows from Theorem~\ref{thm:Reduction}.

\section{Open Problems}
We characterized the automorphism group of the cube $n^d$ for finite $n$ and $d$.
It would be interesting to characterize the automorphisms of the cube with an infinite dimension.
Would the automorphisms be the same even for uncountable dimension?
The lines of the cube $n^d$ cannot be (straightforwardly) generalized to infinite $n$, because of the decreasing coordinate sequences.

We also proved that the {\ColoredCube} problem is basically as hard as the {\sc Graph Isomorphism} problem.
However, for strategy searching algorithms the most important task is to prune the game tree at upper levels, i.e., after constantly many turns.
Thus, a natural question arises: Is there a polynomial time algorithm for the {\ColoredCube} problem if all color classes except one have constant size?

\section*{Acknowledgment}
We would like to thank the anonymous referees for many tips for
improving the presentation and for pointing out the weak spots in the
proofs.
\bibliography{main}
\bibliographystyle{plain}

\end{document}